\documentclass{article}

\usepackage{microtype}
\usepackage{graphicx}
\usepackage{subfigure}
\usepackage{booktabs} 

\usepackage{hyperref}
\usepackage{physics}

\usepackage{amsmath,amsfonts,bm}

\newcommand{\id}{\mathbf{1}}

\def\eqref#1{equation~\ref{#1}}

\def\1{\bm{1}}

\DeclareMathAlphabet{\mathsfit}{\encodingdefault}{\sfdefault}{m}{sl}
\SetMathAlphabet{\mathsfit}{bold}{\encodingdefault}{\sfdefault}{bx}{n}

\DeclareMathOperator*{\argmax}{arg\,max}

\usepackage[accepted]{icml2023}

\usepackage{amsmath}
\usepackage{amssymb}
\usepackage{mathtools}
\usepackage{amsthm}

\usepackage[capitalize,noabbrev]{cleveref}

\theoremstyle{plain}
\newtheorem{theorem}{Theorem}[section]
\newtheorem{prop}[theorem]{Proposition}

\theoremstyle{definition}

\theoremstyle{remark}

\newtheorem{example}{Example}[section]

\usepackage[textsize=tiny]{todonotes}

\icmltitlerunning{The END: An Equivariant Neural Decoder for Quantum Error Correction}

\begin{document}

\twocolumn[
\icmltitle{The END: An Equivariant Neural Decoder for Quantum Error Correction}



\icmlsetsymbol{equal}{*}

\begin{icmlauthorlist}
\icmlauthor{Evgenii Egorov}{equal,uva}
\icmlauthor{Roberto Bondesan}{equal,qcom}
\icmlauthor{Max Welling}{uva}
\end{icmlauthorlist}

\icmlaffiliation{uva}{University of Amsterdam}
\icmlaffiliation{qcom}{Qualcomm AI Research}

\icmlcorrespondingauthor{Evgenii Egorov}{egorov.evgenyy@ya.ru}
\icmlcorrespondingauthor{Roberto Bondesan}{r.bondesan@gmail.com}
\icmlkeywords{Equivariance, Toric code, Quantum Error Correction}

\vskip 0.3in
]



\printAffiliationsAndNotice{\icmlEqualContribution} 

\begin{abstract}
Quantum error correction is a critical component for scaling up quantum computing. Given a quantum code, an optimal decoder maps the measured code violations to the most likely error that occurred, but its cost scales exponentially with the system size.
Neural network decoders are an appealing solution since they can learn from data an efficient approximation to such a mapping and can automatically adapt to the noise distribution.
In this work, we introduce a data efficient neural decoder that exploits the symmetries of the problem. To this end, we characterize the symmetries of the optimal decoder for the toric code and propose a novel equivariant architecture that achieves state of the art reconstruction accuracy compared to previous neural decoders.
\end{abstract}

\section{Introduction}
Quantum computers have exponential advantage compared to classical computers for quantum physics simulations and for breaking certain cryptosystems. They can also provide speed ups for optimization and searching problems. 
However, these quantum advantages are guaranteed only for fault tolerant architectures and quantum error correction is a critical component to build a fault tolerant quantum computer.

The prototypical example of a quantum code is the toric code \cite{Kitaev_2003}, where qubits are placed on the edges of a torus and the logical qubits are associated with operations along the  non-contractible loops of the torus.
This model (or rather its variant with open boundaries) has been implemented in current hardware \cite{krinner2022realizing,PhysRevLett.129.030501,google_surface_code} and is a standard benchmark for developing new decoders.
Recently, alternative quantum LPDC codes have been explored which have better rate at the expense of complicated hardware implementations \cite{quantum_LPDC}.

The decoding problem aims at correcting the errors that occurred in a given time cycle.
Exact optimal decoding is computationally intractable \cite{iyer2013hardness}, and a standard approach in the literature is to devise handcrafted heuristics \cite{Dennis2002,Delfosse2021} that give a good tradeoff between time and accuracy.
The downside of these is however that they are tailored to a specific code or noise model.
Neural decoders have been proposed to overcome these limitations, by learning from data how to adapt to experimental setups.
Neural network decoders also benefit from quantization and dedicated hardware that allow them to meet the time requirements for decoders to be useful when deployed \cite{Overwater_2022}.
Several works therefore studied neural decoders for the toric code.
Pure neural solutions are however limited to small system sizes \cite{Krastanov_2017,wagner2020symmetries} or low accuracy \cite{Ni2020neuralnetwork}. Solutions that combine neural networks with classical heuristics can reach large systems but are limited in their accuracy by the underlying heuristics \cite{meinerz2021scalable}.

Incorporating the right inductive bias in the neural network architecture is an important design principle in machine learning, exemplified by convolutional neural networks, and their generalization, $G$-equivariant neural networks \cite{cohen2016,weiler2021general}
In this work, we show how to improve the performance of neural decoders by designing an equivariant neural network that approximates the optimal decoder for the toric code.
Our contributions are as follows:
\begin{itemize}
    \item We characterize the geometric symmetries of the optimal decoder for the toric code.
    \item 
    We propose an equivariant neural decoder architecture. The key innovation is a novel twisted version of the global average pooling over the symmetry group.
    \item 
    We benchmark a translation equivariant model against neural and non-neural decoders. 
    We show that our model achieves state of the art accuracy compared to previous neural decoders.
\end{itemize}

\section{Related work}

Popular handcrafted decoders for the toric code are the minimum weight perfect matching (MWPM) decoder \cite{Dennis2002} and the union find decoder \cite{Delfosse2021}.
These decoders however treat independently bit and phase flip errors, and they do not count correctly degenerate errors. For these reasons they are practically fast but have limited accuracy.
Not dealing with degenerate errors impacts also their equivariance as discussed in details in Appendix \ref{app:literature}.
Decoders based on tensor network contraction \cite{Bravyi2010,chubb2021general} achieve the highest threshold for the toric code. Their runtime however increases quickly with the bond dimension that controls the accuracy of the approximation and they are difficult to parallelize compared to neural networks. 
Also, contrary to ML methods, they cannot adapt automatically to different noise models.

Several papers have investigated neural networks for quantum error correction, however none of them studies the problem from an equivariance lens.
\cite{Krastanov_2017} uses a fully connected architecture; \cite{wagner2020symmetries} imposes translation invariance by to zero-centering the syndrome and uses a fully connected layer on top;
\cite{Ni2020neuralnetwork} uses a convolutional neural network which does not represent the right equivariance properties of the optimal decoder. Appendix \ref{app:literature} contains details of these architectures and the results obtained in these papers. 
\cite{meinerz2021scalable} obtains the largest system size and threshold among neural decoders by combining a convolutional neural network backbone with a union find decoder. In our work we show that our model, which does not rely on a handcrafted decoder, achieves higher accuracy.

From the perspective of equivariant architectures \cite{cohen2016,weiler2021general}, our work studies a generalized form of equivariance, where the output representation depends on the values of the inputs to the neural network. 
To the best of our knowledge, this type of symmetry properties for a neural network have not been considered before.

Finally, neural decoders for classical error correction were discussed as a form of generalized belief propagation in \cite{satorras2021neural}. However classical and quantum error correction are fundamentally different \cite{iyer2013hardness}, and these results do not directly translate to the quantum case.
See \cite{Liu2019} for an attempt which however does not achieve good accuracy for the toric code.

\section{The toric code}

In this section we review the necessary background on the toric code.
Recall that a qubit $\ket{\psi}$ is a superposition of $0$ and $1$ bits, which are denoted by $\ket{0}$
and $\ket{1}$:
$\ket{\psi} = \alpha\ket{0}+\beta\ket{1}$. 
A quantum error correction code aims at correcting two types of errors on qubits: bit flip errors $X$, and phase flip errors $Z$, which act as: 
$X(\alpha\ket{0}+\beta\ket{1})=\beta\ket{0}+\alpha\ket{1}$ and $Z(\alpha\ket{0}+\beta\ket{1})=\alpha\ket{0}-\beta\ket{1}$. 
We also recall that the space of $n$ qubits is that of superpositions of the $2^n$ possible bit strings of $n$ bits.
We denote by $E_i$ an error that acts only on the $i$-th qubit.  $E_i$ can take four values: $\id, X, Z, XZ$, corresponding to no-error, bit-flip, phase-flip, or combined phase and bit flip.
It turns out that the ability to correct these discrete set of errors is enough to correct general errors.
We refer the reader to \cite{NielsenChuang} for details on quantum error correction.

\begin{figure}[h]
     \centering
         \begin{tikzpicture}
        \draw[step=1, black] 
        (-0.9,-0.9) grid (4.9,4.9);
        \node[] at (-1.1,2) {$0$};
        \node[] at (-1.1,3) {$-1$};
        \node[] at (-1.1,1) {$1$};
        \node[] at (2,5.1) {$0$};
        \node[] at (3,5.1) {$1$};
        \node[] at (1,5.1) {$-1$};
        \draw[blue,thick] (2,4.9) -- (2,-0.9) node[below] {$C_1$};
        \node[blue] at (2,-1.5) {$\bar{Z}_1$};
        \draw[blue,thick] (-.9,2) -- (4.9,2) node[right] {$C_2(\bar{Z}_2)$};
        \draw[red,thick] (-0.9,1.5) -- (4.9,1.5) node[right] {$C_1^*(\bar{X}_1)$};
        \draw[red,thick] (2.5,4.9) -- (2.5,-0.9) node[below] {$C_2^*$};
        \node[red] at (2.5,-1.5) {$\bar{X}_2$};
        \draw[blue,thick] (0,1) -- (1,1);
        \draw[blue,thick] (1,0) -- (1,1) node[above left] {$a$};
        \draw[red,thick] (0.5,0.5) -- (1.5,0.5) -- (0.5,0.5) -- (0.5,-0.5) node[below] {$b$};
        \draw[red, thick] (0.5,4.5) 
        rectangle (-.5,3.5) node[below] {$d$};
        \draw[blue, thick] (1,4) 
        rectangle (0,3) node[below left] {$c$};
        \end{tikzpicture}
        \caption{Toric code square lattice with periodic boundary conditions. Blue paths are $Z$ errors, red paths on the dual lattice are $X$ errors. $c,d$ are $X$- and $Z$-stabilizers, while $C_i,C_i^*$ are logical operators corresponding to non-contracbtle loops around the torus.}
        \label{fig:toric_code}
\end{figure}
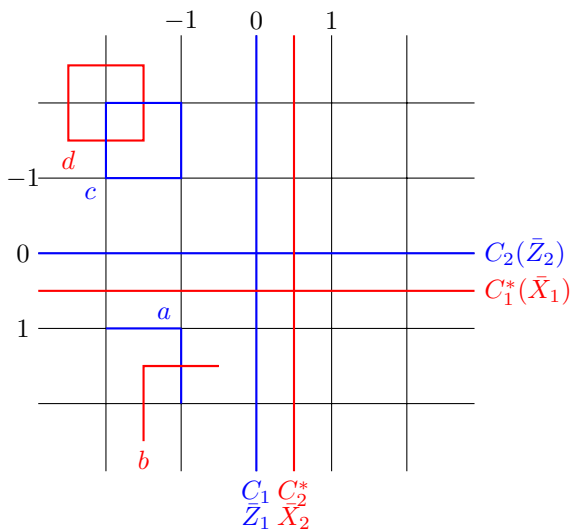

\subsection{Error paths}
The toric code protects against errors by encoding logical qubits in topological degrees of freedom related to the non-contractible cycles of a torus \cite{Kitaev_2003}. This is done as follows.
We start by placing physical qubits on the edges of a $L\times L$ square lattice embedded on a torus. Errors are then associated with paths that traverse the edges corresponding to the qubits affected by errors. 
For reasons that will become clear later, we associate $Z$ errors to paths on the lattice, and $X$ errors to paths on the dual lattice. This is illustrated in figure \ref{fig:toric_code}. Here $a$ represents a $Z$ error on the edges traversed by the paths, while $b$ represents a $X$ errors on the edges traversed by the dual path.

\subsection{Stabilizers and code space}\label{sec:stab_code}
We now consider certain combinations of bit and phase flips called $X$- and $Z$-stabilizers. For each plaquette of the lattice, we define a $Z$-stabilizers as the product of phase flips on the edges around the plaquette. Similarly, for each vertex of the lattice, a $X$-stabilizer is defined as the product of bit flips around a vertex. This is illustrated in figure \ref{fig:toric_code} by the cycles $c$ and $d$.
Note that $Z$-stabilizers are not all independent. In fact if we take the product of two neighbouring plaquettes, the error on the shared edge disappears, since flipping twice is identical to no flipping: $Z^2=\id$. If we take the product of $Z$-stabilizers over all the plaquettes, each edge is counted twice and so all errors disappear. This means that out of the $L^2$ $Z$-stabilizers, only $L^2-1$ are independent. Similarly, for $X$-stabilizers.
The toric code is then defined as the subspace of the $2L^2$ qubits that is preserved by the stabilizer operators. Concreately, if $\ket{\psi}$ is a vector in the $2^{2L^2}$-dimensional space of the physical qubits and $S_i$ a stabilizer, the code subspace is defined by $S_i\ket{\psi}=\ket{\psi}$ for all $i$. 
Note that $S_i^2=\id$ for each stabilizer, so $S_i$ has $\pm 1$ eigenvalues, and imposing the constraint $S_i\ket{\psi}=\ket{\psi}$ reduces the dimension of the space of the qubits by half. Since we have 
$2(L^2-1)$ independent stabilizers, the logical space has dimension $2^{2L^2}/2^{2(L^2-1)}=2^2$, which means that the toric code encodes two logical qubits for any $L$.
We thus see that
an error-free code vector lives in a 4 dimensional vector space. If errors are introduced, this code
vector will develop components in the complement of this code space. The goal of error correction
is to find the most likely projection back onto the code subspace.

\subsection{Logical operators}
We denote logical $X$ and $Z$ operators acting on the logical qubits by $\bar{X}_1,\bar{X}_2,\bar{Z}_1,\bar{Z}_2$. These operators are defined by the paths denoted by $C_1^*,C_2^*,C_1,C_2$ respectively in figure \ref{fig:toric_code}. 
To verify this statement, we need to check the commutation relations of these operators. First, we note that $X$ and $Z$ errors commute if they act on different qubits and anti-commute if they act on the same qubit: $XZ=-ZX$. Thus if we have a $Z$ error string $a$ and a $X$ error string $b$, they commute if they cross an even number of times (so that we have an even number of $-1$'s) and anti-commute if they cross an odd number of times (so that we have an odd number of $-1$'s). For example, the errors represented by the paths $a,b$ in figure \ref{fig:toric_code} anticommute.
We can then check that a $X$-stabilizer always commutes with a $Z$-stabilizer, since they always cross at either $0$ or $2$ edges. 
Similarly, we can check that logical operators commute with stabilizers for a similar reason, but are independent of them -- i.e.~they cannot be written as products of stabilizers -- and thus preserve the logical space but act non-trivially on it, as required to logical operators. Also, we can check that $\bar{X}_i$ anti-commutes with $\bar{Z}_i$ for $i=1,2$ since they cross on a single edge.
We introduce the notation $\omega(E,E')$ to denote whether two errors $E,E'$ anti-commute ($\omega(E,E')=1$) or commute $(\omega(E,E')=0)$.

\section{Symmetries of the toric code decoder}

\subsection{Maximum likelihood decoding}
\label{sec:max_likelihood_decoding}

Let us denote by $p(E)$ the probability for an error $E$ to occur, and assume that $p$ is known.
To correct an (unknown) error $E$ we first measure its syndrome $\sigma$. This is a binary vector of size $2L^2$, whose $i$-th entry is $1$ if 
$E$ anticommutes with the $i$-th stabilizer and zero otherwise.
The decoding problem is then to reconstruct the error given the syndrome.
It is rather easy to produce an error that is compatible with a syndrome. 
In fact, note that syndromes always come in pairs at the end of the error paths, as shown in figure \ref{fig:syndrome} by looking at the error paths $a$ or $c$.
Note that all operators on an error
path, except the ones on the endpoints, intersect twice or zero times a stabilizer, and thus commute with it.
Thus a simple decoding strategy is to return error paths that join syndromes in pairs. Any such paths will produce an error which has the correct syndrome. However, there are many possible errors compatible with a syndrome since both stabilizer and logical operators have trivial syndrome since they commute with all stabilizers. 
For example, the error $c,d$ or $a,b$ have the same syndromes in figure \ref{fig:syndrome}.
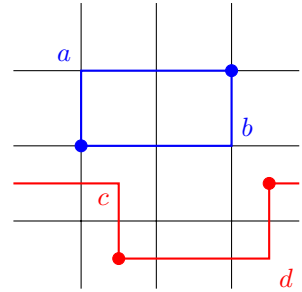
\begin{figure}[h]
         \centering
         \begin{tikzpicture}
        \draw[step=1, black] 
        (-0.9,-0.9) grid (2.9,2.9);
        \draw[red,fill=red] (2.5,0.5) circle (.5ex);
        \draw[red,fill=red] (0.5,-0.5) circle (.5ex);
        \draw[blue,fill=blue] (0,1) circle (.5ex);
        \draw[blue,fill=blue] (2,2) circle (.5ex);
        \draw[blue,thick] (0,1) -- (2,1) node[above right] {$b$} -- (2,2);
        \draw[blue,thick] (0,1) -- (0,2) node[above left] {$a$} -- (2,2);       
        \draw[red,thick] (0.5,-0.5) -- (2.5,-0.5) node[below right] {$d$} -- (2.5,0.5) -- (2.9,0.5);
        \draw[red,thick] (0.5,-0.5) -- (0.5,0.5) node[below left] {$c$} -- (-0.9,0.5);
        \end{tikzpicture}         
      \caption{
      $a$ and $b$ are two possible errors (phase flip paths) that give rise to the same syndrome, here represented by blue dots. Similarly, $c,d$ are two possible errors (bit flip paths) with the same syndrome, represented by red dots.}
     \label{fig:syndrome}
\end{figure}

To understand what constitutes an optimal reconstruction we argue as follows.
First, we note that stabilizer errors do not need to be corrected since by definition they act trivially on the logical qubits, and so two errors $E$ and $E'$ are equivalent if they differ by a stabilizer operator.
However, logical operators do change the logical state, and the optimal decoding strategy is then to choose the most likely logical operator.
The likelihood of a logical operator is to be computed by taking into account that any of the possible errors that are compatible with the syndrome and the logical operator content but differ by a stabilizer could have occurred.

Formally, let us define the vector $L=(\bar{X}_1,\bar{X}_2,\bar{Z}_1,\bar{Z}_2)$.
There are $16$ possible logical operators corresponding to the $4$ binary choices of acting or not with $L_a$, for $a=1,\dots,4$.
Similarly to the syndrome, we define the logical content of an error $E$ as the four-dimensional binary vector $\omega(E,L)$, $L=(\bar{X}_1,\bar{X}_2,\bar{Z}_1,\bar{Z}_2)$.
This allows us to detect whether any of the logical operators are part of $E$. 
(Note that one needs to swap the first two entries of $\gamma$ with the last two entries to reconstruct the logical operator content of $E$ due to the commutation relations. For example, $E=\bar{X}_1$, has $\omega(E,L)=(0,0,1,0)$.)
Then we consider the probability mass of all errors compatible with $\sigma$ and $\gamma$:
\begin{align}
    \label{eq:p_gamma_sigma}
    p(\gamma,\sigma)
    =
    \sum_{E \in {\cal P}}
    p(E) 
    \delta(\omega(E,S),\sigma)
    \delta(\omega(E,L),\gamma),
\end{align}
where ${\cal P}$ is the set of possible errors and $S$ is a vector with all $Z$ and $X$ stabilizers.
From the discussion above, the sum is effectively over all possible $2^{2L^2} $ stabilizer operators -- all the possible products of plaquette and vertex operators.
The most likely $\gamma$ will then allow us to obtain the optimal reconstruction, so maximum likelihood decoding amounts to solving the following optimization problem:
\begin{align}
&\max_{\gamma\in\{0,1\}^{4}} p(\gamma|\sigma)\,.
\end{align}

In the following we shall consider the  i.i.d.~noise called depolarizing noise, which is a standard choice for benchmarking quantum error correction codes \cite{NielsenChuang}:
\begin{align}
\label{eq:depolarizing_noise}
    p(E) = 
    \prod_{e\in {\cal E}}
    \pi(E_e)
    \,,\quad
    \pi
    (E)
    =
    \begin{cases}
    1-p & E = \id\\
    p/3 & E \in \{X,Z,XZ\}
    \end{cases}
    \,.
\end{align}
with ${\cal E}$ the set of edges of the lattice.
The number $p$ is in $[0,1]$ and we give the same probability $p/3$ to the events corresponding to the errors $X,Z,XZ$, while the case of no error has probability $1-3\times p/3 = 1-p$.
    
\subsection{Equivariance properties}
\label{sec:Equivariance properties}

The goal of this section is to derive the equivariance properties of the toric code maximum likelihood decoder.
To start, we define a symmetry of the code as the a transformation $g$ that preserve the group of stabilizers, namely that acts as a permutation of the stabilizers.
Since the logical subspace is defined by $S_i\ket{\psi}=\ket{\psi},~\forall i$, this definition is natural since the logical subspace does not change if we permute the stabilizers.
We call the set of all code symmetries the automorphism
group of the code.

If we denote with prime the transformed quantities, we have 
\begin{align}
    \label{eq:g_S}
    S_i' = S_{gi}\,.
\end{align}
For example, if $g$ is the horizontal translations of the lattice by one unit to the right, it acts on the $Z$ stabilizers $S^Z$'s as:
\begin{align}
    S^Z_p
    =
    \begin{tikzpicture}[baseline={([yshift=-.5ex]current bounding box.center)}]
    \draw[step=1, black] 
    (-0.2,-0.2) grid (2.2,1.2);
    \draw[blue, thick] (0,0) 
    rectangle (1,1);
    \node at (0.5,0.5) {$p$};
    \end{tikzpicture}
    \longrightarrow
    (S^Z)'_p
    =
    \begin{tikzpicture}[baseline={([yshift=-.5ex]current bounding box.center)}]
    \draw[step=1, black] 
    (-0.2,-0.2) grid (2.2,1.2);
    \draw[blue, thick] (1,0) 
    rectangle (2,1);
    \node at (1.5,0.5) {$gp$};
    \end{tikzpicture}
    \,,
\end{align}
and similarly for the $X$ stabilizers $S^X$. We call the set of all code symmetries the automorphism group of the code. 

\begin{figure*} 
  \begin{center}    \includegraphics[width=.9\textwidth]{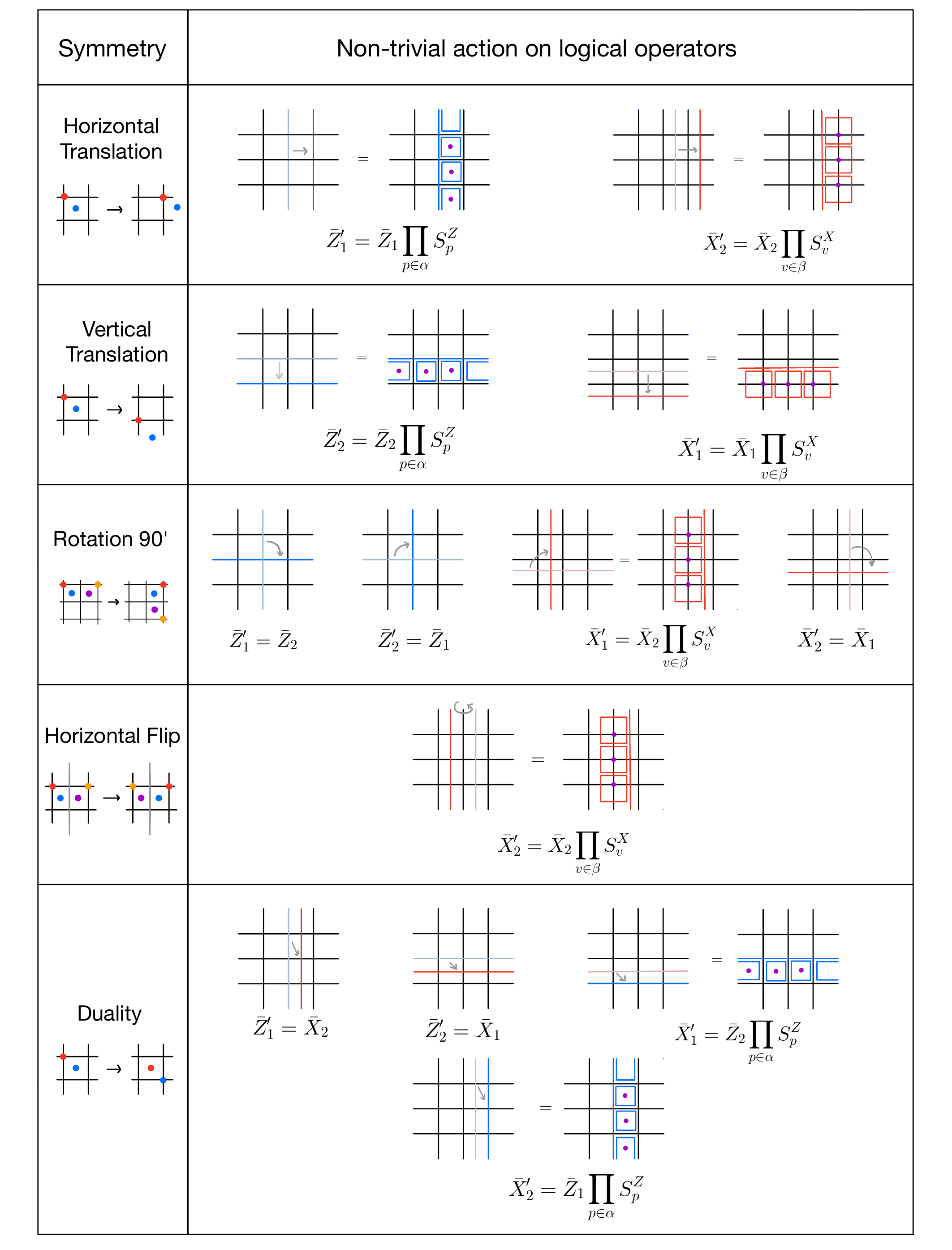}
  \end{center}
  \caption{Left column: the list of symmetries of the toric code decoder and their action on the vertices and plaquettes. Right column: non-trivial action of those symmetries on the logical operators. Purple dots indicate the paths $\alpha,\beta$ in the formulas below the pictures.
  We assume odd $L$ and 
  rotations are performed around the center vertex of the lattice, while horizontal flips are done around the vertical middle line.
  }
\label{fig:table_symmetries}
\end{figure*}

The automorphism group of the toric code is generated by the symmetries of the square lattice, namely horizontal and vertical translations, $90^{\circ}$ rotations and horizontal flips, together with the duality map which switches primal to dual lattice as well as $X$ with $Z$. 
The left column of figure \ref{fig:table_symmetries} shows the action of each of these symmetries on the vertices and plaquettes, defining the permutation of the associated stabilizers.
Logical operators also need to be permuted among themselves up to stabilizers:
\begin{align}
\label{eq:g_L}
L'_a
=
L_{ga}
\prod_{p\in \alpha^g_a}
S_p^Z
\prod_{v\in \beta^g_a}
S_v^X
\,,
\end{align}
where as above $L=(\bar{X}_1,\bar{X}_2,\bar{Z}_1,\bar{Z}_2)$, $ga$ is a permutation of the four elements, and 
$\alpha^g_a$ and $\beta^g_a$ are some $g$-dependent paths on the primal and dual lattice respectively.
The right column of figure \ref{fig:table_symmetries} shows the non-trivial action of the generators of the automorphism group of the toric code on the logical operators.
For example, focusing on the rotation by $90^{\circ}$ row, we see that $ga$ acts as the  permutation $(1234)\to (2143)$.

After discussing the symmetries of the toric code, we now consider the noise distribution. We call a transformation $g$ a symmetry of the noise model if it leaves the noise distribution invariant: $p(E') = p(E)$.
To present the equivariance result, we  find it notationally convenient to see the probability $p(\gamma|\sigma)$ as the $\sigma$-dependent tensor $\bm{p}(\sigma)$ with $4$ indices, $p(\sigma)_{\gamma_1,\gamma_2,\gamma_3,\gamma_4} = p(\gamma_1,\gamma_2,\gamma_3,\gamma_4|\sigma)$.
The permutation part $a\to ga$ for $a=1,2,3,4$ of \eqref{eq:g_L} acts on a tensor 
$t_{\gamma_1,\gamma_2,\gamma_3,\gamma_4}$ 
as the operator $P_g$: 
\begin{align}
    \label{eq:Pg}
    (P_g \bm{t})_{\gamma_1,\gamma_2,\gamma_3,\gamma_4}
    =
    t_{\gamma_{g1},\gamma_{g2},\gamma_{g3},\gamma_{g4}}\,.
\end{align}
With $\alpha_a^g$ and $\beta_a^g$ as in \eqref{eq:g_L} we define the following quantity:
\begin{align}
    \label{eq:Deltasigma}
    (\Delta_g \sigma)_a
    =     
    \sum_{p\in\alpha^g_a}
    \sigma^Z_p
    +    
    \sum_{v\in\beta^g_a}
    \sigma^X_v,    
\end{align}
with $\sigma^Z$ ($\sigma^X$) the syndrome of $S^Z$ ($S^X$).
With these definitions, we are ready to enunciate the equivariance properties of the maximum likelihood decoder.

\begin{theorem}[]\label{thm:symm_decoder}
If $g$ is a symmetry of the code and of the noise model, with action as in \eqref{eq:g_S} and \eqref{eq:g_L},
then the logical probability tensor is invariant under the following action
\begin{align}
    &(\rho_g\bm{p})(\sigma)
    =
    M_g(\sigma)
    \bm{p}(g^{-1}\cdot\sigma)
    \,,\\
    &(g\cdot
    \sigma)_i = \sigma_{gi}\,,\quad
    M_g(\sigma)
    =
    P_g^{-1}
    R^g_1(\sigma)R^g_2(\sigma)R^g_3(\sigma)R^g_4(\sigma)
    \,,
\end{align}
where $R^g_a(\sigma)$ acts as identity if $\Delta_g(g^{-1}\cdot\sigma)_a=0 \mod 2$ and as the flip $t_{\cdots \gamma_a \cdots} \mapsto 
t_{\cdots (1-\gamma_a) \cdots}
$ if $\Delta_g(g^{-1}\cdot\sigma)_a=1\mod 2$.
$P_g$ and $\Delta_g(\sigma)$ are defined in  \eqref{eq:Pg} and \eqref{eq:Deltasigma}.
\end{theorem}
\begin{proof}
    See Appendix \ref{app:proof_symmetry}
\end{proof}

As a corollary of theorem \ref{thm:symm_decoder}, we see that the symmetries of the toric code discussed above (translations, rotations, mirrors and duality) are also symmetries of the maximum likelihood decoder when we have the depolarizing noise of \eqref{eq:depolarizing_noise}.

\begin{example}[]\label{ex:translations}
For concreteness, we here give the explicit formulas for the transformation $M_g(\sigma)$ in the case of  translations.
Referring then to figure \ref{fig:table_symmetries}, if $g$ is the horizontal translation by one unit to the right, then $P_g, R^g_1, R^g_4$ act as identity -- recall that $R^g_1, R^g_4$ are associated to $\bar{X}_1$, $\bar{Z}_2$ which do not change. 
Let us now introduce coordinates on the lattice such that $v = (0,0)$ is the middle vertex (assuming $L$ odd for simplicity), and label other vertices with numbers increasing to the right and bottom,  as in figure \ref{fig:toric_code}.
We also label the plaquette 
neighboring a vertex $(i,j)$ 
to its bottom-right as $(i+\frac{1}{2},j+\frac{1}{2})$. Then we have explicitly,
\begin{align}
    %
    \Delta_g(g^{-1}\cdot\sigma)
    =
    \Big(
    0,
    \sum_{i=0}^{L-1}\sigma^X_{i,0}
    ,
    \sum_{i=0}^{L-1}\sigma^Z_{i+\frac{1}{2},-\frac{1}{2}}
    ,
    0
    \Big)\,.
\end{align}
where the coordinates are understood modulo $L$.
Then $R^g_a$ 
acts as
 $t_{\cdots \gamma_a \cdots} \mapsto 
t_{\cdots (1-\gamma_a) \cdots}
$ if $\Delta_g(g^{-1}\cdot\sigma)_a = 1\mod 2$ and as identity 
$\Delta_g(g^{-1}\cdot\sigma)_a = 0\mod 2$.
Similarly, if $g$ is the vertical translation by one unit to the bottom, we have that $P_g$ is identity and the action of $R^g_a$ is read off from:
\begin{align}
    \Delta_g(g^{-1}\cdot\sigma)
    =
    \Big(
    \sum_{j=0}^{L-1}\sigma^X_{0,j}
    ,
    0,
    0,
    \sum_{j=0}^{L-1}\sigma^Z_{-\frac{1}{2},j+\frac{1}{2}}
    \Big)\,.    
\end{align}
Still referring to figure \ref{fig:table_symmetries}, it is also clear that translations by more than one unit will involve sums over syndromes associated to more than one row or column. For example, if $g$ is the vertical translation by two units to the bottom,
\begin{align}
    \Delta_g(g^{-1}\cdot\sigma)
    =
    \Big(
    \sum_{i=-1}^0
    \sum_{j=0}^{L-1}\sigma^X_{i,j}
    ,
    0,
    0,
    \sum_{i=-1}^0
    \sum_{j=0}^{L-1}\sigma^Z_{i-\frac{1}{2},j+\frac{1}{2}}
    \Big)\,.    
\end{align}
Translating by $L$ units to the bottom or to the right is the same as no translations. In our formalism this follows from the fact that there exists an error $E$ such that:
\begin{align}
\sum_{i,j=0}^{L-1}\sigma^X_{ij}
= 
\sum_{i,j=0}^{L-1}
\omega(E,S^X_{ij})
=
\omega\Big(E,\prod_{ij} S^X_{ij}\Big)
=
0\,.
\end{align}
The first equality is the definition of syndrome, the second uses the fact that $\omega(E,FG)=\omega(E,F)+\omega(E,G)\mod 2$, and the third uses that the product of $X$ stabilizers across all vertices is the identity, as remarked in section \ref{sec:stab_code}. The same argument applies to $\sigma^Z$ and $Z$ stabilizers.

\end{example}

\section{Machine learning approach}

\subsection{Data generation and loss function}
\label{sec:data_loss}

We now set up the task of learning the logical error probabilities $p(\gamma|\sigma)$ 
introduced in section \ref{sec:max_likelihood_decoding}.
The goal is to amortize the cost of maximum likelihood decoding via training a low complexity neural network.

We prepare data as follows. We are given a noise model $p(E)$ from which we can sample errors $E^{1},E^{2},\dots$
Concretely, we shall use below the depolarizing noise of \eqref{eq:depolarizing_noise}, but the arguments of this section hold for any choice of $p(E)$.
We then compute syndrome and logical components associated to each error $E$ as discussed above: 
\begin{align}
\sigma
=
\omega(E,S)
\,,\quad
\gamma
=
\omega(E,L)
\,.
\end{align}
The pairs $(\gamma,\sigma)$'s are distributed according to \eqref{eq:p_gamma_sigma}
and taken to be
inputs and outputs of a supervised learning task. We thus aim at learning a map $\hat{p}$ that maps a syndrome 
$\sigma\in \{0,1\}^{2L^2}$ to a probability distribution over $4$ binary random variables -- one for each $\gamma_a\in\{0,1\}$, $a\in \{1,2,3,4\}$ -- or alternatively over a categorical variable with $2^4 = 16$ values. 
We learn this map by minimizing the cross entropy loss:
\begin{align}
    \mathbb{E}_{\sigma\sim p(\sigma)}
    \mathbb{E}_{\gamma\sim p(\gamma|\sigma)}
    (-\log \hat{p}(\sigma)_{\gamma} )
    \,.
\end{align}
The minimizer of this loss function satisfies
$\hat{p}(\sigma)_{\gamma}=p(\gamma|\sigma)$.
Therefore, we can perform approximate maximum likelihood decoding by taking the maximum over the learnt probabilities.

\subsection{General theory of equivariant architectures}
\label{sec:General theory of equivariant architectures}

Before delving into the neural architecture, we discuss the symmetry action introduced in theorem \ref{thm:symm_decoder}.
Let us suppose that as in the 
theorem \ref{thm:symm_decoder}
we have a vector-valued function $\bm{f}(\sigma)$ 
and a symmetry action
$(\rho_g\bm{f})(\sigma)
=M_g(\sigma)\bm{f}(g^{-1}\cdot\sigma)
$.
For $\rho$ to be a well defined symmetry action (group homorphism) we need $\rho_g\rho_h = \rho_{gh}$ for any $g,h$ in the symmetry group $G$.
As shown in App.~\ref{sec:hom_rho}, this leads to the following relations:
\begin{align}
\label{eq:M relations}
M_{gh}(\sigma)
=
M_g(\sigma)M_h(g^{-1}\cdot \sigma)
\,.
\end{align}

The dependency of $M$ on $\sigma$, the input to the function on which $M$
 acts on, makes the problem more complicated than those typically considered in the machine learning literature on equivariance \cite{weiler2021general}.
In fact, typically one considers 
functions $\bm{f} : V_{\text{in}}\to V_{\text{out}}$, with $V_{\text{in}}, V_{\text{out}}$ input and output linear representations of $G$.
For example, for image classification, $V_{\text{in}}$ is typically the regular representation of the discrete translations group and $V_{\text{out}}$ is the trivial representation.

In our case instead, the output representation matrix $M_g(\sigma)$ depends on the input $\sigma$, and therefore we cannot immediately use the standard theory of equivariant neural networks \cite{weiler2021general}, which prescribes an alternation of layers with different linear representations of the group.
Instead, we solve the problem of parametrizing the invariant function $\bm{p}$ of theorem \ref{thm:symm_decoder} by projecting a general function onto the $G$-invariant subspace by symmetrizing over the group action.
In fact, we use a refinement of this idea that combines it with the standard theory of equivariant neural networks as follows.

\begin{prop}\label{prop:symmetrization}
Consider the group action $(\rho_g\bm{f})(\sigma) = M_g(\sigma)\bm{f}(g^{-1}\cdot \sigma)$ on a function $\bm{f} : \mathbb{R}^d\to \mathbb{R}^\ell$.
If $\bm{\phi} : \mathbb{R}^d \to  \mathbb{R}^{|G|}
\otimes \mathbb{R}^\ell$ is $G$-equivariant,
$\phi_{h,\gamma}(g^{-1}\cdot \sigma)
= \phi_{gh,\gamma}(\sigma)$,
then the following is invariant:
\begin{align}
    \label{eq:G average}
    \hat{\bm{f}}(\sigma)
    =
    \frac{1}{|G|}
    \sum_{h\in G}
    M_h(\sigma)
    \bm{\phi}_{h}(\sigma)
    \,.
\end{align}
\end{prop}
\begin{proof}
    Using \eqref{eq:M relations} and the equivariance hypothesis,
    \begin{align}
    (\rho_g\hat{\bm{f}})(\sigma)
    &=   
    M_g(\sigma)
    \frac{1}{|G|}
    \sum_{h\in G}
    M_h(g^{-1}\cdot \sigma)
    \bm{\phi}_{h}(g^{-1}\cdot \sigma)\\
    &=
    \frac{1}{|G|}
    \sum_{h\in G}
    M_{gh}(\sigma)
    \bm{\phi}_{gh}(\sigma)
    =
    \hat{\bm{f}}(\sigma)\,.
    \end{align}
\end{proof}

We note that the average over the group is the basic principle behind the popular global average pooling layer used at the head of convolutional neural networks.
The key innovation of our construction is to twist the sum by the matrix $M_g(\sigma)$ which ensures the right equivariance.

We summarize here the recipe to build an equivariant neural network for $\bm{p}(\sigma)$ in the case of the translation group.
In this case, $G$ is the product of two cyclic groups of length $L$, $G=\mathbb{Z}_L^{\times 2}$.
Then, elements of the group are indexed by coordinates of the lattice, $g=(i,j)$, and $\bm{\phi}$ is a standard translation-equivariant convolutional neural network
with input the syndrome of size $L\times L\times 2$ and output of size $L\times L\times 16$.
Appendix \ref{app:Implementation details for the translation group} contains details of the implementation of $M_g(\sigma)$ for the translation group, and shows that we can compute the function $\hat{\bm{f}}$ defined in proposition \ref{prop:symmetrization} efficiently in $O(L^2)$ time.

\section{Experiments}
\paragraph{Setup} We benchmark decoders for toric code in the presence of depolarising noise. We compare the performance of our decoder END to the most commonly used non-trainable decoder MWPM \citep{Dennis2002} and the highest-performing neural decoder UFML \citep{meinerz2021scalable}. 

The performance of decoders for lattice size $L$ and physical noise probability $p$ can be measured by the logical accuracy $p_{acc}$, which is the fraction of successfully decoded syndromes over the total number of syndromes. As probability of physical noise increases, logical accuracy decreases for a constant lattice size. Conversely, for a fixed noise level, a bigger lattice produces more accurate results. We can expect that logical accuracy can be expressed as function of lattice size and threshold probability of noise $p_{th}$:
\begin{equation}
    \begin{aligned}
    & p_{acc} = f(L\cdot(p - p_{th})).
    \end{aligned}
    \label{eq:pth}
\end{equation}
Hence, we should compare performance of decoder across several lattice sizes and physical noise probability values. We take the lowest noise probability to be threshold of MWPM decoder ($p_{th}^{\text{\tiny{MWPM}}}=0.155$), as we would like to compare with it. As the highest noise probability we take the highest number from UFML results $0.18$, which is also near theoretical upper bound on threshold $0.188$. We take other two points to be $0.166$ and $0.178$ to be near threshold of the UFML decoder ($p_{th}^{\text{\tiny{UFML}}}=0.167$).

Our decoder needs further clarification since it is trainable. We consider as a decoder a trained model on the particular lattice size and noise level. We then evaluate its performance under various physical noise probabilities and lattice sizes. While bigger lattice size leads to more robust code, for training the number of possible inputs to decoder increases exponentially with lattice size. We test than on lattices sizes 17, 19, 21 as a lattice big enough and practical for  code implementation with physical qubits. We take neural network with the same (up to 10\% difference) number of parameters as UFML.

\paragraph{Results} In Table \ref{tab:acc_decoders} we provide logical accuracy for decoders. We aim to compare END with UFML and MWPM. In the first row we provide upper bound for UFML results, as in original paper they are presented as a figure. For all lattice sizes in range $(7; 63)$ we provide an upper bound for logical accuracy for each noise level. END decoder performers better for each physical noise probability on smaller lattices. Since implementing physical qubits is the challenge, this is significant: one can more robust logical qubit with fewer ($17^2$ instead of $63^2$) by using END as decoder. 

In other blocks of table we provide logical accuracy of the MWPM and END decoder over lattices 17, 19, 21 (each block is sorted in ascending order). All END decoders were trained with noise probability $0.17$ and denoted in the table by $(\text{L}, \text{ch})$: training lattice size and the number of channels in the first block of CNN body. 

For all size of models END decoder outperforms MWPM and UFML decoders. We can speculate about the reasons. Comparing with MWPM both UFML and END decoders can learn correlation between X and Z errors, so this can be a reason of outperforming. Comparing UFML and END decoder we can note that END solves the problem globally: mapping the whole syndrome to the logical state. In contrast, UFML solves problem locally, i.e. given a patch of syndrome provide estimation of the noise realisation per qubit in patch. While this approach scalable in nature it provides worse performance. Following \citep{chubb2021general} we estimated the threshold $p_{th}^{END}$ by evaluation decoder for several lattices on regular grid $[0.145; 0.18]$ of $21$ noise probability values and fit cubic polynomial $f$ over $(p_{acc};L\cdot(p-p_{th})$ pairs from \ref{eq:pth}. Estimated threshold of the END is $0.17$, which is better than UFML ($0.167$) and MWPM ($0.155$), see App. \ref{app:threshold} for plot.

\paragraph{Ablation }In Table \ref{tab:ablation} we show ablation studies where we change the twisted global average pooling of section \ref{sec:General theory of equivariant architectures} for a simple global average pooling or fully connected layer. We found that performance degrades considerably in those cases. Since CNN with average pooling (AP) cannot learn for lattice size 7 and noise level 0.155, it is impossible it will be successful for larger lattices and higher levels of noise. CNN with fully connected head (FC) for lattice size 7 shown performance better than MWPM (and still worse than END), hence we tried to train it on larger lattice and a bigger noise level, where it failed.
\begin{table}[ht]
\defcitealias{meinerz2021scalable}{UFML}
\defcitealias{higgott2021pymatching}{MWPM}
\centering
\scalebox{0.8}{
\begin{tabular}{llrrrr}
\toprule
    Decoder &  $L$ & $p:$ 0.155 &  0.166 &  0.178 &  0.18 \\
\midrule
\citetalias{meinerz2021scalable} & $(7;63)$ & $(0.5;0.6)$ & $<0.6$ & $<0.45$  & $< 0.2$\\
\midrule
\midrule
    \citetalias{higgott2021pymatching}  & 17 &   0.55 &   0.43 &   0.31 &  0.29 \\
 \small{END(17, 32)} & 17 &   0.77 &   0.66 &   0.52 &  0.49 \\
 \small{END(17, 64)} & 17 &   0.82 &   0.72 &   0.57 &  0.55 \\
\small{END(19, 128)} & 17 &   0.82 &   0.72 &   0.58 &  0.55 \\
\bfseries\small{END(17, 128)} & 17 &   \bfseries0.85 &   \bfseries0.75 &   \bfseries0.61 &  \bfseries0.59 \\
\midrule
    \citetalias{higgott2021pymatching}  & 19 &   0.55 &   0.42 &   0.29 &  0.28 \\
 \small{END(17, 32)} & 19 &   0.75 &   0.63 &   0.47 &  0.45 \\
 \small{END(17, 64)} & 19 &   0.82 &   0.70 &   0.54 &  0.52 \\
\small{END(19, 128)} & 19 &   0.84 &   0.72 &   0.57 &  0.55 \\
\bfseries\small{END(17, 128)} & 19 &   \bfseries0.85 &   \bfseries0.74 &   \bfseries0.59 &  \bfseries0.57 \\
\midrule
    \citetalias{higgott2021pymatching} & 21 &   0.55 &   0.41 &   0.28 &  0.26 \\
 \small{END(17, 32)} & 21 &   0.70 &   0.56 &   0.40 &  0.38 \\
 \small{END(17, 64)} & 21 &   0.77 &   0.63 &   0.46 &  0.44 \\
\small{END(17, 128)} & 21 &   \bfseries0.83 &   0.70 &   0.53 &  0.51 \\
\bfseries\small{END(19, 128)} & 21 &   \bfseries0.83 &   \bfseries0.71 &   \bfseries0.55 &  \bfseries0.53 \\
\bottomrule
\end{tabular}}
\caption{\small{Logical accuracy (larger is better) of decoders over depolarising noise for $L: 17, 19, 21$ and noise levels around thresholds of competitive decoders. All \textit{END} decoders were trained with noise probability $0.17$ and $(\text{L}, \text{ch})$ denotes training lattice size and the number of channels in the first block of CNN body. All St.d. $\leq 0.002$, sample size $10^6$. Here UFML is the method of \cite{meinerz2021scalable}  and MPWM the minimum weight perfect matching decoder  \cite{Dennis2002}.
}}
\label{tab:acc_decoders}
\end{table}
\begin{table}[ht]
\scalebox{0.9}{
\begin{tabular}{llrrr}
\toprule
    Ablation &  $L$ & $p:$ 0.155 &  0.166  \\
\midrule
    (7,32,AP) & 7 & 0.13(0.01) & - \\
    (7,32,FC) & 7 & 0.62(0.05) & 0.51(0.05) \\
\midrule
    (15,64,FC) & 15 & - & 0.21(0.02)  \\
    (17,64,FC) & 17 & - & 0.06(0.02)  \\
    (19,64,FC) & 17 & - & 0.1(0.03)  \\
\bottomrule
\end{tabular}}
\caption{Ablation study for \textit{END} decoder. We used same body architecture and training procedure, but the Fully Connected (FC) or Average pooling (AP) projection from feature to the space of logits instead of the equivariant pooling introduced in section \ref{sec:General theory of equivariant architectures}.}
\label{tab:ablation}
\end{table}
The following paragraphs report more technical details about the experiments.
\paragraph{Architecture} We adapt the wide-resnet (WR) architecture \citep{zagoruyko2016wide}: each convolution is defined to have periodic boundaries. WR consists of 3 blocks, where the depth of each block was 3 and fixed across all experiments. We vary the number of channels in the blocks: $(\text{ch},64,64)$, $\text{ch}\in\{32,64,128\}$. Inside each block we used the GeLU \citep{hendrycks2016gaussian} activation function and standard batch-norm. As initialization we used kaiming for leaky ReLU. 
\paragraph{Sampling noise channel} For performance tests of neural decoders we used standard NumPy random generator. During training we used Quasi-Monte Carlo generator based on Sobolev Sequence. This does not provide any gain in terms of performance overall, but we found it to stabilise training. Both for training and performance evaluation batches were generated on the fly.
\paragraph{Training hyperparameters} We used AdamW optimiser \citep{loshchilovdecoupled} for all experiments. In order to avoid manual tuning of schedule and learning rate, we used "1cycle" approach \citep{smith2019super}. Typical maximal learning rate was $0.01$ for batch $512$ and $0.03$ for batch $2048$.
For the ablation studies we also tried reduce on plateau and cosine annealing, however this doesn't produce consistent effects for lattice size bigger than 7.

\section{Conclusions and outlook}
In conclusion, in this work we have shown for the first time how to build neural decoders that respect the symmetries of the optimal decoder for the toric code. We have also benchmarked our novel translation equivariant architecture against other approaches in the literature, finding that our method achieves state of the art reconstruction accuracy compared to previous neural decoders.

Future work will explore implementing other symmetries, scaling up to larger lattices, and deploy the model to interface with a quantum computer.
Our methods can also be extended to other quantum LDPC codes -- where the set of vertices, edges and faces of the square lattice is replaced by more general chain complexes -- and we envision applying equivariant neural decoders to these other codes as well.


\bibliography{refs}
\bibliographystyle{icml2023}

\newpage
\appendix
\onecolumn

\section{Equivariance property of toric code decoders in the literature}
\label{app:literature}
\subsection{Classical decoders}
We first discuss classical, i.e.~non-neural, decoders.

The maximum weight perfect matching decoder (MWPM) is the standard decoder for the toric code \cite{Dennis2002}.
It treats $X$ and $Z$ syndromes independently and returns the minimum (Hamming) weight error consistent with the syndrome, a problem which can be solved using the Blossom algorithm in $O(n^3)$ time, but on average, it takes $O(n)$ time \cite{Fowler2012}. This decoder is popular because of its simplicity, but it has two main drawbacks: first, it treats $X$ and $Z$ error independently; second, it does not account for the degeneracy of errors due to the trivial action of the stabilizer \cite{Poulin2010}.
Here we also show that it does not respect the equivariance properties of the maximum likelihood decoder under translations, see also \cite{wagner2020symmetries} where the authors point out that MWPM is not translation invariant.
We note that the root cause for this failure is the ambiguity of minimum weight decoding for a string of syndromes, which is translation invariant, while the error string returned by the MWPM decoder is not, since it is obtained by breaking the ambiguity by an arbitrary choice, which is not modified after a translation.
Note that degeneracy also can lead to a breaking of the symmetry in the maximum likelihood decoder. In fact, if two logical classes $\gamma,\gamma'$ are such that $p(\gamma|\sigma) = p(\gamma'|\sigma)$ are this value is the largest logical probability, then it does not matter which one we return. This ambiguity can also lead to a non-translation equivariant result of the maximum likelihood decoder.


Now we discuss the union find decoder \cite{Delfosse2021}. Like MWPM, union find also treats $X$ and $Z$ independently -- which leads to suboptimal decisions -- and is a based on a two-stage process: first, during the syndrom validation step errors are mapped onto erasure errors, namely losses that occur for example when a qubit is reinitialized into a maximally mixed state; then, one applies the erasure decoder. The latter simply grows a spanning forest to cover the erased edges starting from the leaves, and flips qubits if it encounters a vertex with a syndrome. The syndrome validation step creates a list of all odd clusters, namely clusters with an odd number of non-trivial syndromes. This is done by growing odd clusters until two meet so that their parity will be even.
We note that the syndrome validation step respects the symmetries of the square lattice as does the erasure decoder. The union find decoder $d$ thus returns a recovery $E$ for a syndrome $\sigma$ so that $d(T \sigma) = Td(\sigma)$ for a translation $T$, leading to the right equivariance expected from a maximum likelihood decoder.
This decoder is also very fast, practically $O(L^2)$, but the heuristics used leads to a suboptimal performance w.r.t.~the MWPM decoder. 

The tensor network decoder achieves state of the art results for the threshold probability of the toric code \cite{Bravyi2014,chubb2021general}. It does so by approximating directly the intractable sum over the stabilizer group that is involved in computing the logical class probabilities. The runtime is $O(n\log n + n\chi^3)$ where $n=L^2$ and $\chi$ is the so-called bond dimension, which is the number of singular values kept when doing an approximate contraction of tensors.
Near the threshold we expect this to grow with the system size, but in practice modest values (e.g.~$\chi=6$ for the surface code in \cite{Bravyi2014} with $L=25$) give good results over a range of noise probabilities.
The symmetries of the decoder will depend on the approximate contraction procedure. Those used in \cite{Bravyi2014,chubb2021general} create a one dimensional matrix product state along a sweep line on a planar embedding of the Tanner graph of the code. 
This procedure breaks the translational invariance of the decoder due to the finite $\chi$, and in these works it was applied only to the surface code, namely the code with boundaries. 
We believe that an equivariant contraction procedure might lead to an even more efficient tensor network decoder.

\subsection{Neural decoders}
\cite{Krastanov_2017} introduces the machine learning problem of predicting the error given a syndrome with a neural network for the toric code. The architecture used is a fully connected network that does not model any symmetries of the toric code. It obtains threshold $16.4$ and studies lattices up to $L=9$.

\cite{wagner2020symmetries} explicitly investigates the role of symmetries for neural decoders. It uses a high level decoder architecture, where an input syndrome $\sigma$ is first passed to a low level decoder which is not learnable and returns a guess for the error, $f(\sigma)$, which will correspond to a given logical class. The syndrome is also passed to a neural decoder that as in our setting predicts the logical probability. This is then combined with the underlying decoder to make a prediction for the error.
In formulas, called the neural prediction $\hat{p}(\gamma|\sigma)$, the logical probabilities returned by the high level decoder is
\begin{align}
\hat{\gamma}(\sigma)
=
\argmax_\gamma
\hat{p}(\gamma|\sigma)
+
\omega(L,f(\sigma))\,.
\end{align}
To take into account symmetries, the authors modify this setup by introduce a preprocessing step to deal with translations and mirror symmetries. For translations for example they define equivalence classes of syndromes related by translations. For each class they define an algorithm that centers the syndrome $\sigma$ to pick a representative, say $[\sigma]$ and pass that as input to both the low level decoder and neural network. By construction, denoted by $T$ the translation operator, one has $[\sigma T]=[\sigma]$. Then the output of the low level decoder is obtained by undoing the translation on the output of the low level decoder on $[\sigma]$. 
Let us call this modified low level decoder $\tilde{f}(\sigma)$ and the translation applied to produce the representative $T_\sigma: [\sigma]= \sigma T_\sigma$.
Then $\tilde{f}(\sigma) = f([\sigma]) T_\sigma^{-1}$ and this means that $\tilde{f}(\sigma)$ is translationally equivariant by construction: $\tilde{f}(\sigma T)=
f([\sigma]) T_{\sigma T}^{-1}
=
f([\sigma]) (T^{-1} T_{\sigma})^{-1}
=
\tilde{f}(\sigma)T$.
The neural network has input $[\sigma]$ and the modified high level decoder used in this paper is:
\begin{align}
\hat{\gamma}(\sigma)
=
\argmax_\gamma
\hat{p}(\gamma|[\sigma])
+
\omega(L,\tilde{f}(\sigma))\,.
\end{align}
Note that we get the correct behavior under translations, see \ref{prop:symm_decoder}:
(Note that $T^{-1}$ appears w.r.t.~\eqref{eq:gamma_prime} since we are considering the equation written as $p(\gamma|\sigma')=p(\gamma''|\sigma)$)
\begin{align}
\hat{\gamma}(\sigma T)
&=
\argmax_\gamma
\hat{p}(\gamma|[\sigma T])
+
\omega(L,\tilde{f}(\sigma) T)
\\
&=
\argmax_\gamma
\hat{p}(\gamma|[\sigma])
+
\omega(LT^{-1},\tilde{f}(\sigma))
\\
&=
\argmax_\gamma
\hat{p}(\gamma|[\sigma])
+
\omega(L,\tilde{f}(\sigma))
+
\rho_{\text{stab}}(T^{-1})\omega(H,\tilde{f}(\sigma))\\
&=
\hat{\gamma}(\sigma T)
+
\rho_{\text{stab}}(T^{-1})\sigma
\,.
\end{align}
While the pipeline proposed in this paper is manifestly equivariant under translations, it requires additional computational cost to preprocess the data, and uses a fully connected network. Further, the authors could only show improvements w.~r.~t.~MWPM decoder for $L=3,5,7$, when using as underlying decoder MWPM itself, which adds additional runtime.

\cite{Ni2020neuralnetwork} implements a neural decoder for large distance toric codes $L=16,64$. The decoder is only tested for bit flip noise, where it performs on par, or lower, to MWPM.
Large distance is achieved by using convolutional layers to downscale the lattice, in a similar fashion to a renormalization group decoder.
The architecture is a a stack of CNN blocks each downsampling by half the lattice size, till the system has size $2\times 2$. Downsampling is done by a convolutional layer with filter size $[2,2]$ and stride $[2,2]$.
The output marginal probabilities for logical classes are then produced by a dense layer on the outputs the CNN blocks: $p(\gamma_0|\sigma),p(\gamma_1|\sigma)$ where $\gamma_i\in\mathbb{F}_2$ corresponds to acting with $\bar{X}_i$ or not.
Note that the marginal probabilities will have a transformation law inherited by that of the joint, namely for translations $\rho_{\text{logi}}(T)=\id$, we have $p(\gamma_i+\rho_{\text{stab}}(T)_{i:}\sigma|\sigma T) = 
p(\gamma_i|\sigma)
$.
The authors did not discuss whether the architecture they propose has this symmetry property.
We conjecture that the architecture in this paper does not have the right symmetry under translations. In fact, we expect that a CNN -- the architecture proposed is a CNN apart from the periodic boundary conditions in the convolutions -- can approximate only a translation invariant function, in our case $p(\gamma_i|\sigma T) = 
p(\gamma_i|\sigma)$, and a function with the equivariance properties required by the actual logical probabilities.

\cite{meinerz2021scalable} uses a CNN backbone which processes patches of the lattice to produce the probability that the central qubit of the patch has an error, and then adds on top a union find decoder to deal with correlations beyond the size of the patch that the neural network sees. Using a CNN (and assuming periodic padding), the system is equivariant under translations, and so is the union find decoder, so the whole procedure amounts to a decoder $d(T \sigma) = Td(\sigma)$ for a translation $T$, leading to the right equivariance expected from a maximum likelihood decoder.
While relying on the union find decoder for long range correlations allows one to scale to large lattices (up to $L=255$), it also limits its accuracy, which leads to a threshold probability of $0.167$.

\section{Proof of theorem \ref{thm:symm_decoder}}
\label{app:proof_symmetry}

To prove theorem \ref{thm:symm_decoder}, we shall first establish the following proposition which shows the transformation of the components of the maximum likelihood decoder.

\begin{prop}[]\label{prop:symm_decoder}
If $g$ is a symmetry of the code and noise model, then for all $\gamma\in\mathbb{F}_2^{2k}, \sigma\in \mathbb{F}_2^{n-k}$, we have $p(\gamma| \sigma)=p(\gamma'| \sigma')$, with
\begin{align}
    \label{eq:gamma_prime}
    \sigma' &= 
    g^{-1}\cdot  \sigma\\
    \label{eq:gamma prime app}
    \gamma' 
    &= 
    \rho^{-1}_{\text{logi}}(g)
    (
    \gamma + 
    \Delta_g(\sigma')
    \mod 2
    )     
    \,,
\end{align}
where $\rho_{\text{logi}}$ is the permutation representation of the logical operators in \eqref{eq:g_L} and $\Delta_g$ is defined in \eqref{eq:Deltasigma}.
\end{prop}
\begin{proof}
If we denote by $\pi(g)$ the action of a symmetry on the error $E$, since $\omega(E,FG) = \omega(E,F)+\omega(E,G) \mod 2$, and $p(E)=p(\pi(g)E)$ by assumption, we have $\omega(\pi(g)E,\pi(g)F)=\omega(E,F)$, so:
\begin{align*}
    p(\gamma, \sigma)
    &=
    \sum_{E\in {\cal P}}
    p(E) 
    \, \delta( \omega(E, S), \sigma)
    \delta(\omega(E, L), \gamma)\\
    &=
    \sum_{E\in {\cal P}}
    p(\pi(g) E) 
    \, 
    \delta( \omega(\pi(g) E, \pi(g)S), \sigma)
    \delta(\omega(\pi(g) E,\pi(g)L ), \gamma)\\
    &=
    \sum_{E\in {\cal P}}
    p(E) 
    \, 
    \delta( \omega(E, \pi(g)S), \sigma)
    \delta(\omega(E,\pi(g)L ), \gamma)
    \\
     &=
    \sum_{E\in {\cal P}}
    p(E) 
    \, 
    \delta( \omega(E, S), 
    g^{-1}\cdot \sigma)
    \delta(
    \rho_{\text{logi}}(g)
    \omega(E,L)
    +
    \Delta_g(\omega(E,S)), \gamma)\\
    &=p(\gamma',\sigma')
    \,.
\end{align*}
In the third to last equality we relabeled $\pi(g)E$ with $E$ since $\pi(g)$ is an invertible transformation on the set of Pauli operators and thus acts as a permutation of the Pauli errors.
In the second to last equality we used the transformation laws of $S$ and $L$, \eqref{eq:g_S} and \eqref{eq:g_L}.

The probability $p(\gamma|\sigma)$ has the same symmetry since $p(\gamma|\sigma)=p(\gamma,\sigma)/p(\sigma)$ and the denominator $p(\sigma)$ is invariant: $p(\sigma)
=\sum_{\gamma}p(\gamma,\sigma)
=\sum_{\gamma}p(\gamma',\sigma')
=\sum_{\gamma'}p(\gamma',\sigma')
=
p(\sigma')$.
\end{proof}

The theorem \ref{thm:symm_decoder}  follows by noting that the map $p(\gamma,\sigma)\mapsto p(\gamma',\sigma')$ can be written as the operator 
$P^{-1}_gR_1R_2R_3R_4$ defined in theorem \ref{thm:symm_decoder}
acting on the tensor $\bm{p}(\sigma)$.
Indeed the map $p(\gamma_1,\gamma_2,\gamma_3,\gamma_4) \mapsto 
p(\rho_{\text{logi}}(g)\gamma)=
p(\gamma_{g1},\gamma_{g2},\gamma_{g3},\gamma_{g4})$ can be written as $P_g$ acting on the tensor 
$\bm{p}$, with $P_g$ explicitly in Dirac notation:
\begin{align}
    P_g 
    =
    \sum_{\gamma} \ket{\gamma_1,\gamma_2,\gamma_3,\gamma_4}\bra{\gamma_{g1},\gamma_{g2},\gamma_{g3},\gamma_{g4}}
    =
    \sum_{\gamma} \ket{\gamma_{g^{-1}1},
    \gamma_{g^{-1}2},
    \gamma_{g^{-1}3},\gamma_{g^{-1}4}}\bra{\gamma_{1},\gamma_{2},\gamma_{3},\gamma_{4}}
    \,.
\end{align}
This is the same object introduced in \eqref{eq:Pg}. It is a representation of the symmetric group of $4$ elements: $P_g P_h = P_{gh}$.
The map $p(\gamma) \mapsto 
p(\gamma + \Delta_g(\sigma) \mod 2)$
can be written as the following operator acting on tensor $\bm{p}$:
\begin{align}
    \bigotimes_{a=1}^4
    \left[\delta_{\Delta_g(\sigma)_a,0}\id_2 +
    \delta_{\Delta_g(\sigma)_a,1}X
    \right]\,,
\end{align}
with $X$ the Pauli $X$.
This proves the form of the operator 
$P_g^{-1}R_1\cdots R_4$
introduced in \ref{thm:symm_decoder}.

\section{Group homorphism property of the representation $\rho$}
\label{sec:hom_rho}

If we denote $\bm{f}'(\sigma) = (\rho_h\bm{f})(\sigma) = M_h(\sigma)\bm{f}(h^{-1}\cdot \sigma)$, the condition $\rho_g\rho_h = \rho_{gh}$, means that we have:
\begin{align}
(\rho_g\rho_h\bm{f})(\sigma)
=
M_g(\sigma)\bm{f}'(g^{-1}\cdot \sigma)
=
M_g(\sigma)M_h(g^{-1}\cdot \sigma)
\bm{f}(h^{-1}g^{-1}\cdot \sigma)\,,
\end{align}
which needs to equal $M_{gh}(\sigma)\bm{f}((gh)^{-1}\cdot \sigma)$, that is
\begin{align}
M_{gh}(\sigma)
=
M_g(\sigma)M_h(g^{-1}\cdot \sigma)
\,.
\end{align}
This is a necessary condition for $\rho$ to be a well defined action.

\section{Implementation details for the translation group}
\label{app:Implementation details for the translation group}

We shall now discuss some details of the construction of proposition \ref{prop:symmetrization} for the translation group.
We index elements of the translation group $\mathbb{Z}_L^{\times 2}$ as $g=(i,j)$ indicating a translation to the right by $i$ and to the bottom by $j$. $\bm{\phi}$ is then a standard translation-equivariant convolutional neural network: 
\begin{align}
\phi((-i,-j)\cdot \sigma)_{k,l,\gamma}= 
\phi(\sigma)_{k+i,l+j,\gamma}
\,.    
\end{align}
From \eqref{eq:M relations} with $g=(i,0)$, $h=(0,j)$, we have 
\begin{align}
M_{(i,j)}(\sigma) 
&= M_{(i,0)}(\sigma) 
M_{(0,j)}((-i,0)\cdot\sigma)\\
&=
M_{(i,0)}(\sigma) 
M_{(0,j)}(\sigma)
\end{align}
where the second equality follows from the fact that $M_{(0,j)}(\sigma)$ depends on $\sigma$ only through sums along rows which are invariant under horizontal translations.
We can then consider the horizontal and vertical translations separately.
Setting $g=(i-1,0)$ and $h=(1,0)$ in
\eqref{eq:M relations} we get a recursion relation 
\begin{align}
M_{(i,0)}(\sigma) = 
M_{(i-1,0)}(\sigma)
M_{(1,0)}((-i+1,0)\cdot \sigma)
\,.    
\end{align}
We discussed explicitly $M_{(1,0)}(\sigma)$ in example \ref{ex:translations}. $M_{(1,0)}((i,0)\cdot \sigma)$ involves the sum 
of the syndrome over the $i+1$-th column of vertices or plaquettes -- when starting counting from the middle, as in figure \ref{fig:toric_code} -- and can be precomputed for all $i$ by summing along the columns of the matrices $\sigma^X$ and $\sigma^Z$.
Therefore we can compute $M_{(i,0)}(\sigma)$ from $M_{(i-1,0)}(\sigma)$ in $O(1)$ time.
A similar procedure allows us to compute 
$M_{(0,j)}(\sigma)$ so that
the summation in \eqref{eq:G average} can be computed efficiently in $O(L^2)$.

Since our experiments focus on the translation symmetry, we refrain from discussing here details of the implementation of the other symmetries of section \ref{sec:Equivariance properties}.

\section{Threshold plot}
\label{app:threshold}
\begin{figure}[h]
    \centering
    \includegraphics[scale=0.5]{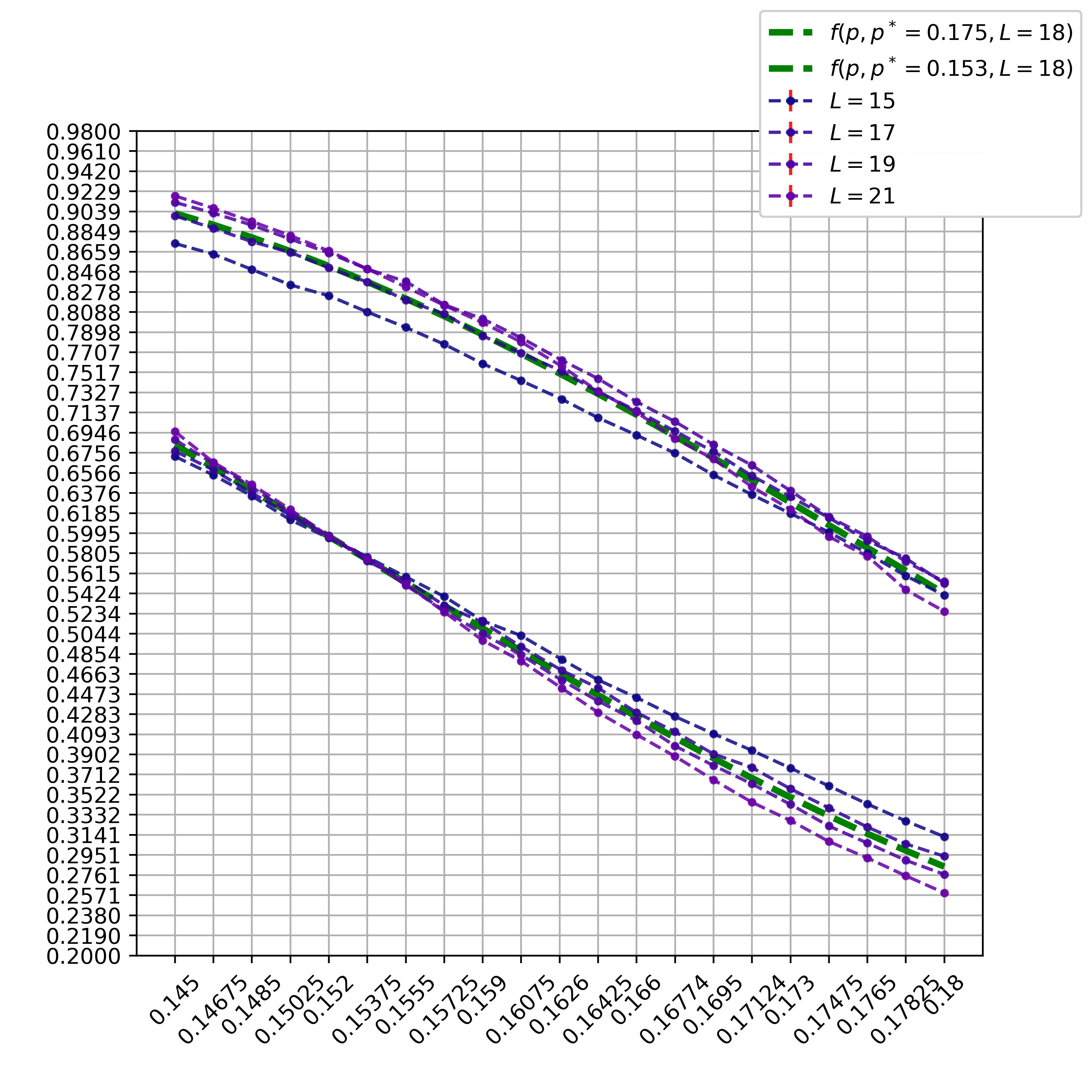}
    \label{fig:pth}
    \caption{Logical accuracy of MWPM and \textit{END} decoder. \textit{END} decoder has threshold $\approx 0.17$, MWPM $\approx 0.155$}
    \label{fig:threshold}
\end{figure}


\end{document}